\documentclass{IEEEtran}

\newcommand{\cJ}{\mathcal J}

\newcommand{\cN}{\mathcal N}
\newcommand{\cO}{\mathcal O}

\newcommand{\cX}{\mathcal X}

\newcommand{\bbE}{\mathbb{E}}

\newcommand{\bbP}{\mathbb{P}}

\newcommand{\bbR}{\mathbb{R}}

\newcommand{\bmat}{\begin{bmatrix}}
\newcommand{\emat}{\end{bmatrix}}

\newcommand{\bsmat}{\begin{bsmallmatrix}}
\newcommand{\esmat}{\end{bsmallmatrix}}

\newcommand\undermat[2]{%
  \makebox[0pt][l]{$\smash{\underbrace{\phantom{%
    \begin{matrix}#2\end{matrix}}}_{\text{$#1$}}}$}#2}

\usepackage{bm}
\usepackage{commath}
\usepackage{comment}
\usepackage{mathtools}
\usepackage{algorithm}
\usepackage[noend]{algpseudocode}
\usepackage{amsmath,amsthm,amssymb,amsfonts,mathrsfs}
\allowdisplaybreaks

\usepackage{float}
\usepackage{xcolor}
\usepackage{dcolumn}
\usepackage{caption}
\usepackage{makecell}
\usepackage{multirow}
\usepackage{graphicx}
\usepackage{subfigure}
\usepackage{subcaption}
\usepackage{tabularray}
\usepackage{stackengine}

\usepackage{url}
\usepackage{cite}
\usepackage{soul}
\usepackage{bigfoot}
\usepackage{textcomp}
\usepackage{enumitem}
\usepackage{anyfontsize}
\usepackage[euler]{textgreek}
\usepackage[hidelinks]{hyperref}
\DeclareNewFootnote[para]{default}

\theoremstyle{plain}
\newtheorem{lemma}{Lemma}

\newtheorem{theorem}{Theorem}

\newtheorem{proposition}{Proposition}

\theoremstyle{definition}

\newtheorem{problem}{Problem}
\newtheorem{definition}{Definition}
\newtheorem{assumption}{Assumption}

\theoremstyle{remark}

\def\BibTeX{{\rm B\kern-.05em{\sc i\kern-.025em b}\kern-.08em
    T\kern-.1667em\lower.7ex\hbox{E}\kern-.125emX}}

\begin{document}
\title{Identifying Time-varying Costs in Finite-horizon Linear Quadratic Gaussian Games}
\author{Kai Ren and Maryam Kamgarpour	\thanks{Ren and Kamgarpour are with the SYCAMORE Lab, École Polytechnique Fédérale de Lausanne (EPFL), Switzerland (e-mail: {\tt\small kai.ren@epfl.ch; maryam.kamgarpour@epfl.ch}).}
\thanks{Ren's research is supported by Swiss National Foundation Grant $\#200020\_207984 \slash  1$.}} 

\maketitle 
\begin{abstract}
    We address cost identification in a finite-horizon linear quadratic Gaussian game. We characterize the set of cost parameters that generate a given Nash equilibrium policy. We propose a backpropagation algorithm to identify the time-varying cost parameters. We derive a probabilistic error bound when the cost parameters are identified from finite trajectories. We test our method in numerical and driving simulations. Our algorithm identifies the cost parameters that can reproduce the Nash equilibrium policy and trajectory observations. 
\end{abstract}

\begin{IEEEkeywords}
Game theory; Identification; Stochastic control.
\end{IEEEkeywords} \vspace{-3mm}

\section{Introduction} \label{sec:introduction}
\IEEEPARstart{I}{n} real-world multi-robot applications, such as warehouse automation and autonomous cars, agents operate in shared environments and influence each other’s decisions. Dynamic games \cite{Basar1998} provide a framework for modeling multi-agent interactions, where agents jointly control a dynamical system and optimize their own costs. Open-loop \cite{cleach2021algames, Zhu2023} and feedback Nash equilibria \cite{Laine2023, Zhong2023, Ren2025} capture the control input sequences or policies where no agents can unilaterally reduce their costs. These frameworks have been applied in autonomous driving and robotics. 

Solving the Nash equilibrium typically requires the cost functions of all players, which are often unknown. For example, at an unsignalized intersection as shown in Fig.~\ref{fig:noiseTraj}, an aggressive green vehicle that puts higher weights on the goal-reaching objective may force others to yield, while a cautious one allows others to proceed first. Inferring the underlying costs of other agents enables the prediction of interactions under varying scenarios or dynamics.

Cost identification in a single-agent setting has a long history since the work of Kalman~\cite{Kalman1964WhenIA}. Identifying cost from control policies has been investigated in both continuous-time~\cite{Priess2015} and discrete-time~\cite{Menner2018ECC} inverse linear quadratic control. For a given policy, the corresponding cost matrices are characterized through the algebraic Riccati equation, and their computation is formulated as a semi-definite program. Cost identification from the state \cite{ZHANG2019IOC} or output \cite{Zhang2019ILQR} trajectories has been addressed in a discrete-time finite-horizon inverse linear quadratic control. These works formulate a convex program from Pontryagin’s Maximum Principle characterization of the optimal control to compute the cost matrices. 

In the multi-agent setting, \cite{Inga2019, Huang2022} characterize the set of cost matrices that induce a feedback Nash equilibrium in an infinite-horizon linear quadratic differential game. The characterization is based on the coupled Riccati equations. These works, however, focus on stationary policies and cost matrices in the infinite-horizon setting. 

In realistic multi-robot applications~\cite{Ren2025, Zhong2023}, finite-horizon and time-varying costs are commonly adopted. For example, at an interaction, a driver focuses more on avoiding collisions when two cars are close. As they move apart, the driver shifts attention toward reaching the destination.

\begin{figure}[t]
    \centering
    \includegraphics[width=0.8\linewidth]{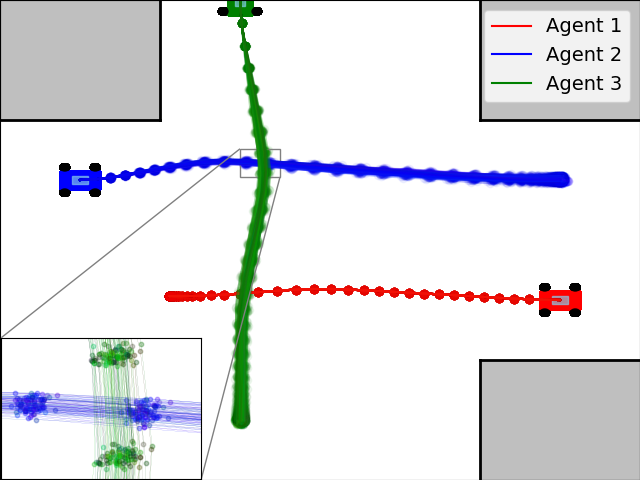}
    \caption{Demonstration trajectories of three cars driving in a cross intersection. The trajectories are generated by the same Nash equilibrium policy but different system noise realizations.}
    \label{fig:noiseTraj}
\end{figure}

In the finite-horizon setting, considering nonlinear dynamics or noisy, partial observations, past work has developed approaches beyond Riccati or Pontryagin-based characterizations. In discrete-time nonlinear dynamic games, maximum likelihood estimation has been used to identify costs from noisy state and input trajectories corresponding to feedback~\cite{Li2023} and open-loop~\cite{peters2021inferringobjectives, Liu2023} Nash equilibria. In differential games, cost identification from open-loop Nash equilibria~\cite{Molloy2017, Awasthi2020} is achieved by minimizing the residual of the optimality condition. Bayesian methods have been developed for autonomous driving, where the Gaussian \cite{cleach2020lucidgamesonlineunscentedinverse} or discrete \cite{peters2023contingencygames} distributions of cost parameters are updated using online state and input observations.

An open question is the characterization of the time-varying cost parameters in a finite-horizon linear quadratic Gaussian (LQG) game, given a Nash equilibrium policy or trajectories generated by the Nash policy. We consider a cost identification framework for a finite-horizon LQG game. Our main contributions are
\begin{itemize}
    \item We characterize the set of time-varying cost parameters that yield a given Nash equilibrium policy, and present a backpropagation algorithm for cost identification.
    \item We develop a probabilistic upper bound on the cost identification error when the Nash equilibrium policy is estimated from finite trajectories.
    \item We validate our method with numerical and driving simulations, showing that our method can identify the cost parameters that reconstruct the Nash equilibrium policies and trajectories in multi-agent interactions.
\end{itemize}

\textit{Notations:} A Gaussian distribution with mean $\mu$ and covariance matrix $\Sigma$ is denoted as $\mathcal{N}(\mu, \Sigma)$. We denote a set of consecutive integers by $[a]=\{1, 2, \dots, a\}$ and $[a]^- = \{0, 1, \dots, a-1\}$.  The weighted \(L_2\) norm of a vector \( x \) with a weighting matrix \( P \succeq 0\) is $\|x \|_{P} = \sqrt{x^\top P x}$. Let $I_n$ denote the identity matrix in $\bbR^{n\times n}$. We use $0_{m\times n}$ and $1_{m\times n}$ to represent the matrices of 0 and 1 in $\bbR^{m\times n}$, respectively. The vectorization of $A \in \mathbb{R}^{m \times n}$ is denoted by $\operatorname{vec}(A) \in \mathbb{R}^{mn}$. The Kronecker product of $A \in \mathbb{R}^{m \times n}$ and $B \in \mathbb{R}^{p \times q}$ is denoted by $A \otimes B \in \mathbb{R}^{mp \times nq}$. We use $\wedge$ to denote a logical AND operator. Let $f(\cdot), g(\cdot)$ be real-valued functions, we write $f(\epsilon) = \mathcal{O}(g(\epsilon))$ if there exist constants $c > 0$ and $\epsilon_0$ such that $|f(\epsilon)| \le c\,|g(\epsilon)|$ for all $\epsilon \geq \epsilon_0$.

\section{Problem Formulation}
We consider the cost identification from the following linear quadratic Gaussian (LQG) game: \vspace{-3mm}

{\small
\begin{subequations} \label{problem:unconstrainedLQ} 
\begin{align} 
&\underset{\gamma^i}{\text{min}}
&& \cJ^i(\gamma) := \bbE\left[\sum_{t=0}^{T-1} \left( \|x_{t+1}\|^2_{Q^{i}_{t+1}} + l^{i\top}_{t+1} x_{t+1} + \|u^i_t\|^2_{R^{i}_{t}} \right)\right],  \label{eq:LQcost}\\[-6pt]
&\text{s.t. }
&&   x_{t+1} = A_t x_t + \sum_{i=1}^{N} B_{t}^{i} u^i_t + \omega_t, \label{eq:dynamics} \\
&&&  x_0 \sim \cN(\mu_0, \chi_0), \quad w_t \sim \cN(0, \Sigma_t). \label{eq:Gaussian}
\end{align}
\end{subequations}}

We denote the state and control input of agent \( i \) at time \( t \) as \( x^{i}_{t} \in \mathbb{R}^{n_{x^{i}}}\) and \( u^{i}_{t} \in \mathbb{R}^{n_u} \), respectively. The shared state vector \( x_t = [x_{t}^{1^\top}, \ldots, x_{t}^{N^\top} ]^\top  \in \mathbb{R}^{n_x} \), with $n_{x} = \sum_{i=1}^{N} n_{x^{i}}$,  contains the states of all agents. The system dynamics matrices are $A_t \in \mathbb{R}^{n_x \times n_x}$ and $B_t^i \in \mathbb{R}^{n_x \times n_u}$ are the system's dynamics matrices. The cost weighing matrices are $Q^{i}_{t} \in \mathbb{R}^{n_x \times n_x}$, $l^{i}_{t} \in \mathbb{R}^{n_x}$ and $R^{i}_{t} \in \mathbb{R}^{n_u \times n_u}$. We assume $R_t^i$ is diagonal and $R^{i}_{t} \succ 0$, which is realistic in robotic applications, as the costs penalize the magnitude of the inputs (e.g., acceleration, angular velocity) and these inputs are typically decoupled.  

We consider a Gaussian initial state $x_0 \sim \cN(\mu_0, \chi_0)$ and a system noise $w_t \in \mathbb{R}^{n_x}$ to be a sample of a Gaussian distribution $w_t \sim \cN(0, \Sigma_t)$.

A feedback policy for player $i$ is denoted as $\gamma^i= (\gamma^i_0, \gamma^i_1, \ldots, \gamma^i_{T-1})$, where $\gamma^i_t: \cX \rightarrow \bbR^{n_u}$ determines the control input based on the evolving states, $u_t^i = \gamma^i_t(x_t)$. We denote \( \gamma^{-i} \) the control inputs of all players except player \( i \).
\begin{definition}
A feedback Nash equilibrium is a control policy $(\gamma^{i*}, \gamma^{-i*})$ where for every player $i \in [N]$, $\cJ^i(x_0, (\gamma^{i*}, \gamma^{-i*})) \leq \cJ^i(x_0, (\gamma^i, \gamma^{-i*})), \; \forall \gamma^i.$
\end{definition}

Under suitable conditions\footnote{The invertibility of $\Phi_t, \;\forall t\in[T]^-$, defined by \eqref{eq:phi} in Appendix~\ref{appendix:phi}.} \cite[Remark 6.5]{Basar1998}, a unique affine state-feedback Nash equilibrium for the LQG game \eqref{problem:unconstrainedLQ} exists and is given by 
\begin{equation} \label{eq:feedbackU}
    \gamma^{i*}_{t}(x_t) = -K^{i*}_{t} x_t - \alpha^{i*}_{t}, \;\;\; \forall t \in [T]^-,
\end{equation} \vspace{-5mm}

Observe that in contrast to past works in inverse linear quadratic regulator~\cite{Priess2015, Menner2018ECC} and inverse linear quadratic game~\cite{Inga2019}, we explicitly incorporate a linear term $l^{i\top}_{t+1} x_{t+1}$ in the cost. This modeling choice is important in robotics applications, where linear cost terms can encode reference tracking or local approximation of collision cost \cite[Lemma 1]{Ren2025}. For example, consider a reference tracking penalty of $\|x_t - \bar{x}_t\|^2_2 = x_{t}^\top Q_{t} x_{t} + l^{\top}_{t} x_{t} + c$. A different reference point $\bar{x}_t$ yields the same $Q_t = I_{n_x}$ but different $l_t = -2 \bar{x}_t$. \vspace{2mm}

In this section, we aim to address the following problem.
\begin{problem} \label{problem:1}
    Given the feedback Nash equilibrium policies $\{K^{i*}_{t}, \alpha^{i*}_{t}\}_{t \in [T]^-}^{i \in [N]}$ of the LQG game \eqref{problem:unconstrainedLQ}, determine the cost parameters $\{Q^{i}_{t+1}, l^{i}_{t+1}, R^{i}_{t}\}_{t \in [T]^-}^{i \in [N]}$, whose Nash equilibrium coincides with the given policies.
\end{problem}

\subsection{Cost identification in finite-horizon LQG games}
We first characterize the set of time-varying cost parameters that yield a Nash equilibrium policy in a finite-horizon LQG game \eqref{problem:unconstrainedLQ}. We enforce the following assumption for the characterization.
\begin{assumption}\label{assump:CLinverse}
 The closed-loop system \(F_t = A_t - \sum_{j=1}^N B_t^j K^{j*}_t\) is invertible for all $t \in [T]^-$. 
\end{assumption}

Assumption~\ref{assump:CLinverse} can be verified a priori from the known dynamics and Nash policies. Analogous assumption has been proven~\cite[Theorem~2.1]{ZHANG2019IOC} in a single-agent inverse linear quadratic setting.

The characterization is based on a vectorization and parametrization of the coupled Riccati equation, similar to the formulations in the single-agent linear quadratic regulator \cite{Menner2018ECC} and the infinite-horizon linear quadratic game \cite{Inga2019}. The following matrices are defined for all $t \in [T]^-$, based on he known dynamics and the Nash equilibrium policy:
 \begin{subequations} \label{eq:defNoBack}
  \begin{align}
    &S_t^i = F_t^\top \otimes (B_t^i)^\top \in \bbR^{n_un_x \times n_x^2}, \\
    &\textstyle E_{t}^i = K_{t}^{i*} F_{t}^{-1} \sum_{j=1}^N B_{t}^j \alpha^{j*}_{t} + \alpha^{i*}_{t}\in \bbR^{n_u}. 
\end{align}
 \end{subequations}

\begin{proposition} \label{prop:characterization}
Under assumption~\ref{assump:CLinverse}, for a given Nash equilibrium policy $\{K^{i*}_{t}, \alpha^{i*}_{t}\}_{t \in [T]^-}^{i \in [N]}$, the cost parameter of player \( i \) at time \( t \) 
\[
\theta_t^i := 
[
\textnormal{vec}(Q_{t}^i)^\top, \;
\textnormal{vec}(l_{t}^i)^\top,  \;
\textnormal{vec}(R_{t-1}^{i})^\top,  \;
1]^\top \in \bbR^L,\]
is characterized by 
\begin{equation}
    \label{eq:vec}
    M_t^i \theta_t^i = 0_{n_u(n_x+1)  \times 1},
\end{equation} 
{\small
\begin{align*} 
    & \text{\normalsize where \;} M_t^i = \begin{bmatrix}
        S_{t-1}^i & 0_{n_un_x \times n_x} & - K_{t-1}^{i* \top} \otimes I_{n_u} & S_{t-1}^i\bar{\Delta}_{t}^i \\
        0_{n_u \times n_x^2} & B_{t-1}^{i \top} &  - E_{t-1}^{i \top} \otimes I_{n_u} &  B_{t-1}^{i \top}\bar{\Omega}_{t}^i
    \end{bmatrix}. \label{eq:vecDef}
\end{align*}}

The terms $\bar{\Delta}_{t}^i$ and $\bar{\Omega}_{t}^i$ are defined recursively backward. \vspace{1mm}

At $t = T$, set  \(\bar{\Delta}_{T}^i = 0_{n_x^2}, \; \bar{\Omega}_{T}^i = 0_{n_x},\) from which $M_T^i$ is determined for all $i \in [N]$. Using \eqref{eq:vec}, we obtain $\{Q_{T}^i, l_{T}^i, R_{T-1}^i\}^{i\in [N]}$, and we set \(P^{i*}_T = Q_{T}^{i}, \; \zeta^{i*}_T = l_T^i\) for all $i \in [N]$. \vspace{2mm}

For $t = T-1, \dots, 1$, the terms $\bar{\Delta}_{t}^i$ and $\bar{\Omega}_{t}^i$ are updated as follows.  
{\small
\begin{subequations} \label{eq:recur}
\begin{alignat}{2}
&\bar{\Delta}_{t}^i 
= \textnormal{vec}\!\left( F_{t}^{\top} P^{i*}_{t+1} F_{t} + K_{t}^{i*\top} R_{t}^i K_{t}^{i*} \right), \\
&\textstyle \bar{\Omega}_{t}^i 
\textstyle = \textnormal{vec}\!\left[
F_{t}^\top \!\left(\zeta_{t+1}^{i*} 
- P_{t+1}^{i*} \sum_{j = 1}^N B_{t}^j \alpha_{t}^{j*}\right)
+ (K^{i*}_{t})^\top R_{t}^i \alpha^{i*}_{t}
\right].
\end{alignat}
\end{subequations}}

We obtain $M_t^i$ for all $i \in [N]$.  
Using \eqref{eq:vec}, we solve for $\{Q_{t}^i, l_{t}^i, R_{t-1}^i\}^{i\in [N]}$, and update  
{\small
\begin{align*}
P^{i*}_t &= F_t^\top P^{i*}_{t+1} F_t + (K^{i*}_t)^\top R_t^i K^{i*}_t + Q_{t}^{i}, \\
\zeta^{i*}_t & \textstyle = F_t^\top \!\left( \zeta^{i*}_{t+1} - P^{i*}_{t+1} \sum_{j = 1}^N B_t^j \alpha_t^{j*} \right)
+ (K^{i*}_t)^\top R_t^i \alpha^{i*}_t + l^i_t.
\end{align*}}

\end{proposition}

The proof of Proposition~\ref{prop:characterization} is provided in Appendix~\ref{appendiex:unconstrainedILQG}. Proposition~\ref{prop:characterization} states that for each player $i$ and timestep $t$, the cost parameter $\theta_t^i \in \bbR^L$ that can reproduce the Nash equilibrium policy lies in the null space of $M_t^i$. Here, $L := n_x^2 + n_x + n_u^2 + 1$ and $M_t^i \in \mathbb{R}^{(n_u n_x + n_u) \times L}$. The null space of $M_t^i$ is non-empty, as $n_u n_x + n_u < L$.

In contrast to the infinite-horizon setting~\cite{Inga2019}, where one set of coupled Riccati equations determines the stationary policy, the finite-horizon case requires propagating \eqref{eq:recur} backward in time to capture the time-varying cost structure. In particular, characterizing $\theta^i_t$ requires $\bar{\Delta}_{t}^i, \bar{\Omega}_{t}^i$, which depend on $\theta^i_{t+1}$.

\subsection{Cost identification algorithm}
 Given the characterization in Proposition~\ref{prop:characterization}, we solve the following constrained least-squares problem to compute the cost parameters for certain $i$ and $t$, as was done in~\cite{Inga2019} for the infinite horizon setting.
\begin{subequations} \label{eq:QP}
    \begin{alignat}{2}
        & \min_{\theta^i_t} \quad r_t^i = \|M_t^i \theta^i_t\|_2^2, \label{residual}\\
        & \;\text{s.t. } \quad D \; \theta^i_t \;\ge\; \tau. \label{eq:nontrivialConst}
    \end{alignat}
\end{subequations}
where $\tau$ is a small positive value and the inequality \eqref{eq:nontrivialConst} holds element-wise. \( D\) is a selector matrix that enforces strict positiveness only to the unknown entries in \( \theta_t^i \). 
Specifically, \( D \) excludes the off-diagonal elements of \( R_t^i \) that are 0 and the last element of \( \theta_t^i \) that is 1.  The constraint \eqref{eq:nontrivialConst} ensures non-trivial cost matrices (i.e., non-zero and $R^i_t \succ 0$). The feasible region is convex, as \eqref{eq:nontrivialConst} represents hyperplanes.
 
The process for computing the time-varying cost parameters over a finite horizon is summarized in the following algorithm.

\begin{algorithm}[ht]
\begin{minipage}{0.5\textwidth}
\caption{Cost identification in finite-horizon LQG game}
\label{alg:ILQGG}
\begin{algorithmic}[1]
\State \textbf{Input:} Nash equilibrium policy: $\{K^{i*}_{t}, \alpha^{i*}_{t}\}^{i \in [N]}_{t \in [T]^-}$.

\State \textbf{Initialize:}  $\bar{\Delta}_{T}^i = 0_{n_x^2}, \; \bar{\Omega}_{T}^i = 0_{n_x},$  for all $i \in [N]$.

\For{\(t = T, T-1, ..., 1\)} \vspace{1mm}

\State $\{\theta_t^{i}\}^{i \in[N]} \leftarrow$ Solve \eqref{eq:QP} for all $i \in [N]$; \label{algStep:solveNE} \vspace{1mm}
\State $\{\bar{\Delta}_{t-1}^i, \bar{\Omega}_{t-1}^i\}^{i \in[N]} \leftarrow$ \eqref{eq:recur} for all $i \in [N]$;
\EndFor
\State \Return $\{Q^{i}_{t+1}, l_{t+1}^{i}, R_t^{i}\}^{i \in [N]}_{t \in [T]^-}$
\end{algorithmic}
\end{minipage}
\end{algorithm}

We will apply Algorithm~\ref{alg:ILQGG} to both numerical and driving examples in Section~\ref{sec:exp}, demonstrating that the identified cost parameters can recover the Nash equilibrium policies.

\section{Finite-sample probabilistic cost identification error bound} \label{section:learningError}
In Problem~\ref{problem:1}, we assume that the Nash equilibrium policy is known exactly. In practice, the policy may not be available. Instead, demonstration trajectories are often accessible from historical data, as shown in Fig.~\ref{fig:noiseTraj}. Let us denote a collection of trajectories by $\{X_t, U_t^{i}\}_{t \in [T]}^{i \in [N]}$, which consists of $n$ samples: \(X_t =[x_t^{(1)}, x_t^{(2)}, \cdots, x_t^{(n)}]\), \(U_t^i = [u_t^{i(1)}, u_t^{i(2)}, \cdots, u_t^{i(n)}]\). Across demonstrations, all agents follow the Nash equilibrium policy~\eqref{eq:feedbackU}, while $(x_t^{(s)}, u_t^{i(s)})$ vary for $s \in [n]$ due to the system noise in~\eqref{eq:dynamics}. We address the following problem in this section.

\begin{problem} \label{problem:2}
    Given $n$ trajectories $\{X_t, U^{i}_t\}^{i \in [N]}_{t\in[T]}$ generated by a control policy $\{K^{i*}_{t}, \alpha^{i*}_{t}\}_{t \in [T]^-}^{i \in [N]}$, determine the cost parameters $\{Q^{i}_{t+1}, l^{i}_{t+1}, R^{i}_{t}\}_{t \in [T]^-}^{i \in [N]}$, whose Nash equilibrium coincides with the original control policy.
\end{problem}

To address problem~\ref{problem:2}, we first estimate the Nash equilibrium policy from finite trajectories, which yields $\{\widehat{K}^{i*}_{t}, \widehat{\alpha}^{i*}_{t}\}^{i \in [N]}_{t \in [T]^-}$. The estimated policies are then used in Algorithm~\ref{alg:ILQGG} and yields $\{\widehat{\theta}_t^i\}_{t \in [T]^-}^{i \in [N]}$. Our goal is to upper bound the cost identification error $\|\widehat\theta_t^i-\theta_t^i\|$ given finite state and input trajectory samples. 
To this end, we first characterize the policy estimation error $\|[\widehat{K}^{i*}_{t}, \;\widehat{\alpha}^{i*}_{t}] - [K^{i*}_{t},\;\alpha^{i*}_{t}] \|$ in section \ref{sec:learnTraj}. We then provide the cost identification error bound in section~\ref{sec:learnError} by leveraging properties of the constrained least-squares problem \eqref{eq:QP}.

\subsection{Policy estimation error bound from finite trajectories} \label{sec:learnTraj}

Given $n$ trajectories, the policy parameters can be found through the following linear regression problem:
\begin{equation} \label{eq:regression}
   [\widehat{K}^{i}_t, \; \widehat{\alpha}^{i}_t] = \arg \min_{[K^{i}_t, \, \alpha^{i}_t]} \left\| 
   \begin{bmatrix}
        K^{i}_t, & \alpha^{i}_t
   \end{bmatrix}
   \begin{bmatrix}
        X_t \\ \mathbf{1}_{1\times n} 
   \end{bmatrix} + U_t^{i} \right\|_2^2.
\end{equation}

In practice, inputs (e.g., acceleration) are difficult to measure accurately. Hence, we consider accurate state observation and noisy input observations:
\begin{equation} \label{eq:noiseInput}
    u_t^{i(s)} = -K_t^{i*} x_t^{(s)} - \alpha_t^{i*} + \nu^{(s)}_t.
\end{equation}

The input observation noise $\nu_t^{(s)}$ is assumed to be a sample from a sub-Gaussian distribution with $\bbE(\nu_t^{(s)}) = 0$, which is a mild assumption that includes both Gaussian noise and noise with bounded support. 

Note that under noiseless input observations $\nu_t^{(s)} = 0$, each $\{x_t^{(s)}, u^{i(s)}_t\}$ provides an exact point on the hyperplane \(u_t^{i} = -K_t^{i*} x_t - \alpha_t^{i*}\). Hence, $\eqref{eq:regression}$ is equivalent to solving a system of linear equations with $(n_x + 1)n_u$ unknowns. The policy can be identified exactly with at least $(n_x + 1)n_u$ linearly independent state samples for all $t\in[T]^-$ and $i\in [N]$.

Problem \eqref{eq:regression} is a standard linear regression problem, for which a finite-sample probabilistic error bound exists. 

\begin{lemma}[\cite{krikheli2018finitesampleperformancelinear}] \label{lemma:policyConcentration} 
    With a sample size of $n(\epsilon, \delta) = \mathcal{O}(\epsilon^{-2} \log(\delta^{-1}))$, for every $t \in [T]^-$ and $i \in [N]$, 
\begin{equation} \label{eq:paramConcentration}
    \mathbb{P}\left( \|[\widehat{K}^{i*}_{t}, \;
        \widehat{\alpha}^{i*}_{t}] - [K^{i*}_{t}, \;
        \alpha^{i*}_{t}] \| \leq \epsilon \right) \geq 1 - \delta, 
\end{equation}
where $n(\epsilon, \delta)$ is given in~\cite[Theorem~3.1]{krikheli2018finitesampleperformancelinear} and provided in Appendix~\ref{appendix:nDef} for completeness.
\end{lemma}

\subsection{Cost identification error bound from finite trajectories} \label{sec:learnError}
In this subsection, we study how the probabilistic error \eqref{eq:paramConcentration} in the estimated policy  propagates to the identified parameters $\{\widehat{\theta}^i_t\}^{i \in [N]}_{t \in [T]^-}$ found in Algorithm~\ref{alg:ILQGG}. 

Let the active set in~\eqref{eq:QP} denote the indices of constraints that are satisfied with equality \(\mathcal{A}_t^i(\theta_t^{i}) := \big\{ j \,\big|\, (D\theta_t^{i})_j = \tau \big\}.\) We introduce the following assumptions for analyzing the cost identification error bound $\|\widehat{\theta}^{i}_t - \theta^{i}_t\|$. 

\begin{assumption} \label{assump:activeSet}
The active set of \eqref{eq:QP} remains unchanged. In particular, \(\mathcal{A}_t^i(\widehat{\theta}_t^i) = \mathcal{A}_t^i(\theta_t^{i}), \; \forall t \in [T]^-, \; i \in [N].\)
\end{assumption}

Assumption~\ref{assump:activeSet} is standard in sensitivity analysis \cite{Daniel1973, Lotstedt1983}. It prevents discontinuous changes in the optimal solution due to the change of the feasible region of \eqref{eq:QP}.

Next, we present our main result, which establishes a probabilistic bound on the identification error of the cost parameters obtained from the estimated Nash equilibrium policies.

\begin{theorem}\label{theorem:costBound}
Under assumption~\ref{assump:activeSet}, and given $n \ge n(\epsilon, \delta)$ trajectories as per Lemma~\ref{lemma:policyConcentration}, the following bound holds for the cost parameter $\theta_t^{i}$ identified from the exact Nash policy and its estimate $\widehat{\theta}_t^i$ from $n$ trajectories: \vspace{-2mm}

{\small
\begin{equation} \label{eq:boundforT}
\mathbb{P}\!\left(\|\widehat{\theta}_t^i - \theta_t^{i}\| = \mathcal{O}(\epsilon)\right) \ge (1 - \delta)^{T-t+1}, \;\;\; \forall t \in [T],\; i \in [N].
\end{equation}}
\end{theorem}

Before presenting the proof, we provide some insights to Theorem~\ref{theorem:costBound}. The cost identification error scales linearly with the policy estimation error, ensuring that small policy inaccuracies do not lead to huge errors in the identified cost parameters. Observe that the probability bound in~\eqref{eq:boundforT} degrades over time. The reason is that errors accumulate through the finite-horizon backpropagation in Algorithm~\ref{alg:ILQGG}, since the identification of each $\theta_t^i$ depends on $\theta_{t+1}^i$. Nevertheless, with an arbitrarily small $\delta$, the bound could hold with high probability.

\begin{proof}
First, it is straightforward to show that given \eqref{eq:paramConcentration},
\begin{align}
   & \|\widehat{K}^i_t - K^{i*}_t\| = \mathcal{O}(\epsilon) \wedge \|\widehat{\alpha}_t^{i} - \alpha_t^{i*}\| = \mathcal{O}(\epsilon)  \label{eq:paramBound}
\end{align}
holds with probability $1-\delta$ for every $t \in [T]^-$ and $i \in [N]$. We can then show that for every $t \in [T]^-$ and $i \in [N]$ (the proofs of \eqref{eq:paramBound} and \eqref{eq:rightImply} are provided in Appendix~\ref{appendix:exProof}):
{\small
\begin{align}
\left[ \eqref{eq:paramBound} \wedge \|\widehat{\theta}_{t+1}^i - \theta_{t+1}^{i}\| = \mathcal{O}(\epsilon) \right] \Rightarrow \|\widehat{M}_t^i - M_t^i\| = \mathcal{O}(\epsilon). \label{eq:rightImply}
\end{align}}

Under assumption~\ref{assump:activeSet}, if $\operatorname{rank}(M_t^i) = \operatorname{rank}(\widehat{M}_t^i)$, for all $t \in [T]^-$ and $i \in [N]$~\cite[Theorem~4]{Lotstedt1983}:
\begin{equation} \label{eq:MthetaBound}
   \|\widehat{M}^{i}_t - M^{i}_t\| = \cO(\epsilon) \; \Rightarrow \; \|\widehat\theta_t^i-\theta_t^i\| \; = \; \cO(\epsilon).
\end{equation}

We now prove the probabilistic bound \eqref{eq:boundforT} by induction. 

\textit{Base case.}  
At $t = T$, $\bar{\Delta}_T^i = \bar{\Omega}_T^i = 0$, $M_{T}^i$ only depends on $\{K^{i*}_{T-1}, \alpha_{T-1}^{i*}\}^{i \in [N]}$ and dynamics. Given \eqref{eq:paramBound} holds with probability $1-\delta$, we have \(\bbP \left(\|\widehat{M}_{T}^i - M_{T}^i\| = \mathcal{O}(\epsilon)\right) \ge 1 - \delta\). By~\eqref{eq:MthetaBound}, we have \(\bbP\left(\|\widehat{\theta}_{T}^i - \theta_{T}^{i}\| = \mathcal{O}(\epsilon)\right) \ge 1 - \delta\). \vspace{1mm}

\textit{Induction step.}  
For a given $t \in [T-1]$, let us suppose \(\bbP\left(\|\widehat{\theta}_{t+1}^i - \theta_{t+1}^{i}\| = \mathcal{O}(\epsilon)\right) \geq (1 - \delta)^{T-t}\). Furthermore, as shown above \eqref{eq:paramBound} holds with probability $1-\delta$. It follows from \eqref{eq:rightImply} that \(\bbP\left(\|\widehat{M}_t^i - M_t^i\| = \mathcal{O}(\epsilon)\right) \ge (1 - \delta)^{T-t+1}.\) By~\eqref{eq:MthetaBound}, we have \(\bbP\left(\|\widehat{\theta}_t^i - \theta_t^{i}\| = \mathcal{O}(\epsilon)\right)\ge(1 - \delta)^{T-t+1}\), completing the induction.
\end{proof}

\section{Case studies} \label{sec:exp}
In this section, we aim to a) test whether Algorithm~\ref{alg:ILQGG} can identify cost parameters that reproduce given control policies, and b) empirically examine how the number of demonstrations affects the accuracy of the recovered policy and trajectories from the identified cost parameters. To address a), we test our method in a numerical example (Section~\ref{sec:num}), and to address b), we identify costs in a multi-vehicle driving scenario with varying numbers of demonstrations (Section~\ref{sec:driving}).

\subsection{Numerical example} \label{sec:num}
\begin{figure}[t]
    \centering
    \includegraphics[width=\linewidth]{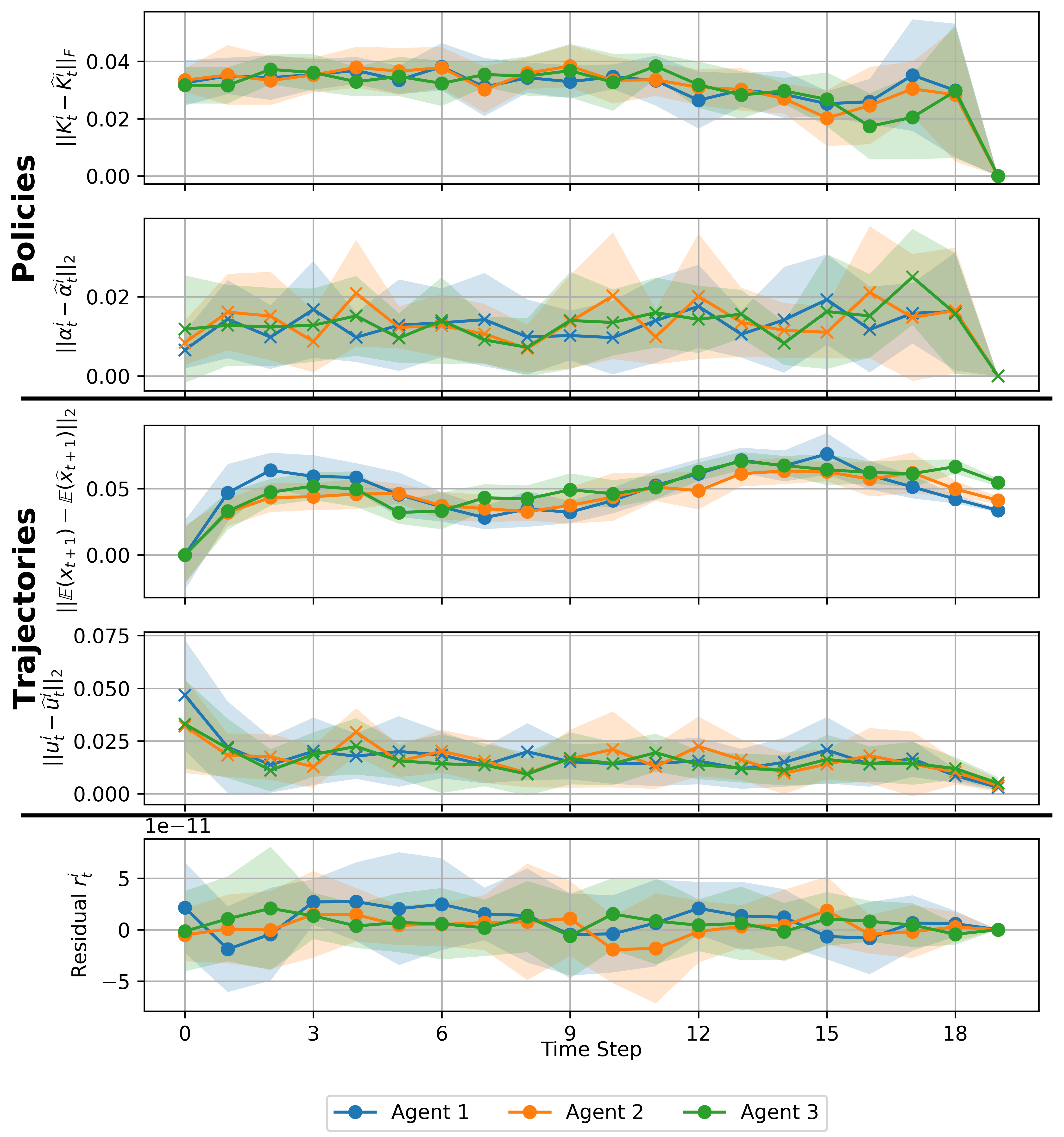} \vspace{-3mm}
    \caption{Optimization residual and the deviations between the recovered (from the identified costs) and true policy, state, and input trajectories across 100 randomized cost matrices.}
    \label{fig:norm_diff_scalar}
\end{figure}

We conducted a numerical simulation of an LQG game \eqref{problem:unconstrainedLQ}. The game is defined with $N=3$ players and states $x_t \in \bbR^3$. The initial state $x_0$ is drawn from a Gaussian distribution \eqref{eq:Gaussian} with randomly generated $\mu_0$. The input for each player is $u^i_t \in \bbR$, and the time horizon is $T=20$. We choose dynamic matrices for all $t \in [T]^-$
\[
A_t = \begin{bmatrix}
    1 & 0 & 0\\
    0 & 1 & 0\\
    0 & 0 & 1
\end{bmatrix}, \quad B^1_t = \begin{bmatrix}
    1 \\
    0 \\
    0
\end{bmatrix}, \; B^2_t = \begin{bmatrix}
    0 \\
    1 \\
    0
\end{bmatrix}, \; B^3_t = \begin{bmatrix}
    0 \\
    0 \\
    1
\end{bmatrix}.
\]

\setcounter{figure}{3}
\begin{figure*}[t] \vspace{-5mm}
    \centering
    \includegraphics[width=0.97\linewidth]{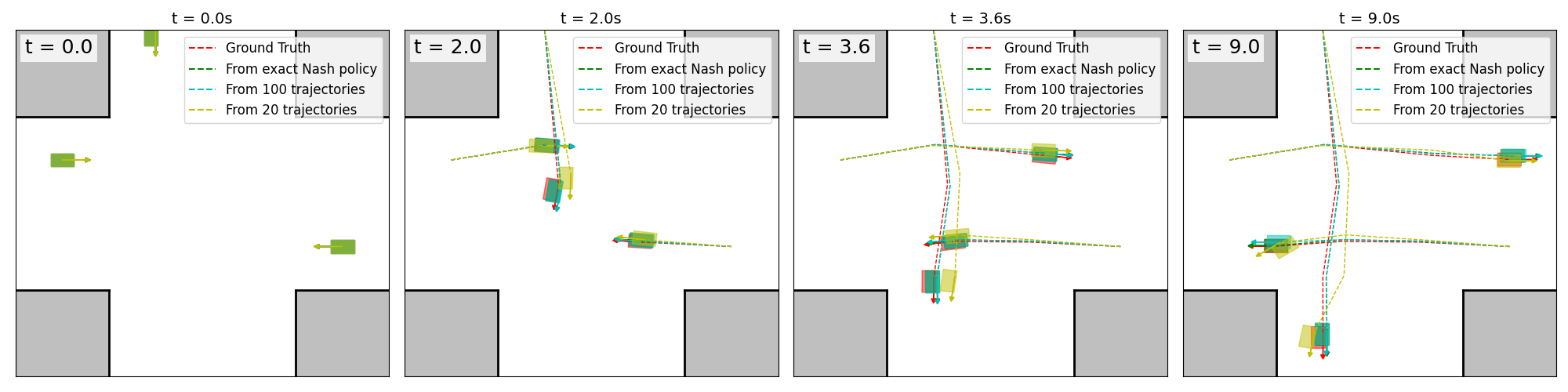} \vspace{-1mm}
    \caption{Expected trajectories generated by the true costs and the identified costs from the exact Nash policy, 100 demonstrations and 20 demonstrations, respectively. The trajectories recovered from the exact policy or 100 demonstrations are close to the ground-truth. With 20 demonstrations, the recovered trajectories become inaccurate.}
    \label{fig:trajectory_car}
\end{figure*}

We generated 100 randomized LQG games, where in each game, the cost matrices $Q_t^i, R_t^i$ are diagonal with random positive entries. The linear cost vector is $l_t^i = \rho_1 [1, -1, 0] + \rho_2[0, 1, -1]$, with randomly positive weights $\rho_1, \rho_2$.

\setcounter{figure}{2}
\begin{figure}[t]
    \centering
    \includegraphics[width=0.9\linewidth]{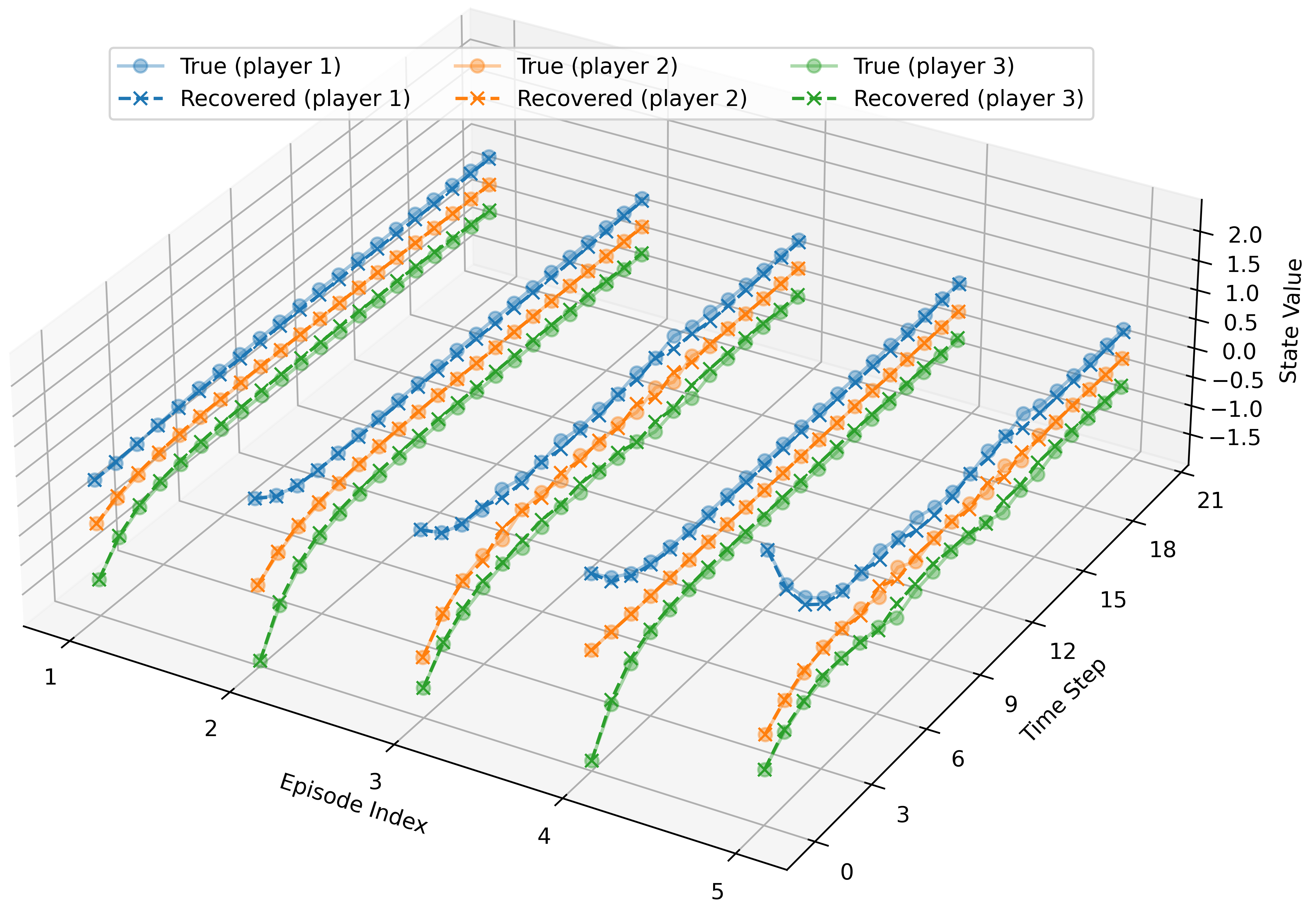}
    \caption{Five exemplary episodes for section~\ref{sec:num}, where we compare the trajectories recovered from the identified cost parameters with the ground-truth. The recovered trajectories closely follow the ground-truth in all episodes.}
    \label{fig:trajectories_num}
\end{figure}

In each game, we first obtain the Nash equilibrium policy. Then, we use Algorithm~\ref{alg:ILQGG} to identify the cost parameters from the Nash policy. From the identified costs, we recover the policy and trajectories. Figure~\ref{fig:norm_diff_scalar} illustrates the optimization residuals defined in~\eqref{residual}, and the norm differences between the recovered and the true policy, state, and input trajectories. The optimization residuals remain near zero, up to negligible numerical errors, indicating that the identified parameters lie in the solution set \eqref{eq:vec}. Furthermore, the recovered policy parameters, states, and control inputs closely match the ground truth across 100 randomized episodes with varying cost matrices.

Fig.~\ref{fig:trajectories_num} presents five examples of the recovered three-player trajectories obtained from the identified cost parameters, compared against the ground-truth. The recovered trajectories align closely with the true ones, showing that the identified costs successfully reproduce the true multi-agent behaviors.

\subsection{Multi-vehicle driving scenario} \label{sec:driving}
In this subsection, we apply our cost identification method to a three-vehicle interaction scenario at a cross intersection~\cite[Section~IV.A]{Ren2025}. We compare the performance of cost identification from varying number of demonstrations.

As shown in Fig.~\ref{fig:trajectory_car}, at $t=0$, three vehicles arrive at an unsignalized intersection and aim to cross as quickly as possible while avoiding collisions. This interaction can be modeled as an LQG game~\eqref{problem:unconstrainedLQ}, where collision avoidance is considered as a linear cost via local approximation \cite[Lemma 1]{Ren2025}. We first solve the Nash equilibrium policies of all vehicles.

Given the Nash equilibrium policies, we generate demonstration trajectories under different realizations of system noise~\eqref{eq:Gaussian} and input observation noise~\eqref{eq:noiseInput}. We construct datasets containing 20 and 100 demonstrations. We then identify cost parameters from the exact policy, 20 samples, and 100 samples, respectively. When using demonstration data, the policies are first estimated via~\eqref{eq:regression}. Algorithm~\ref{alg:ILQGG} is then applied to identify the costs, from which the corresponding control policies and trajectories are reconstructed. 

Table~\ref{tab:performance} summarizes the performance of cost identification across different sample sizes. In particular, we measure the mean (for example, \(\bar{\alpha} = 1/(NT) \sum_{t=0}^{T-1}\sum_{i=1}^N \|\widehat{\alpha}^i_t - \alpha_t^{i*}\|_2\)) and standard deviation of the error between the reconstructed (from the identified costs) and the Nash equilibrium policies. With access to the exact policy or 100 demonstrations, both the policy and trajectory errors remain small. When using only 20 demonstrations, the errors increase substantially. This indicates a small sample size degrades cost identification and trajectory reconstruction accuracy.

\begin{table}[t]
\caption{Mean and standard deviation of the policy and trajectory errors across varying sample size.} \vspace{-1mm}
\centering
\begin{tabular}{| l | r | r | r | r | r |}
    \hline
     & \makecell[r]{\vspace{-1mm} \textbf{$\bar{K}$}} & $\bar{\alpha}$ & $\bar x$ & \(\bar{u}\) \\[0.5ex]  \hline
    Exact policy & 0.17 $\pm$ 0.07  & 0.05 $\pm$ 0.01 & 0.07 $\pm$ 0.05 & 0.02 $\pm$ 0.03 \\ \hline
    $n = 100$ & 0.19 $\pm$ 0.15 & 0.14 $\pm$ 0.07 & 0.10 $\pm$ 0.04 & 0.13 $\pm$ 0.07  \\ \hline
    $n = 20$ &0.28 $\pm$ 0.08  & 0.82 $\pm$ 0.52 & 0.47 $\pm$ 0.37 & 0.73 $\pm$ 0.52 \\ \hline
\end{tabular}
\label{tab:performance} 
\end{table}

Figure~\ref{fig:trajectory_car} compares the expected trajectories (under $\omega_t = 0$ for all $t \in [T]^-$) generated by the reconstructed policies with the ground truth. The trajectories recovered from the exact policy and from 100 samples closely match the ground truth. The trajectories recovered from 20 samples deviate noticeably, indicating that the policy was not accurately recovered.

\section{Conclusion}
We studied the cost identification problem in a finite-horizon linear quadratic Gaussian game. We characterized the set of cost parameters that generate given Nash equilibrium policies and developed optimization-based algorithms for identifying time-varying costs. We derived probabilistic bounds on the cost identification error when the Nash equilibrium policy is estimated from finite trajectories. Through numerical and driving simulations, we demonstrated that the proposed method recovers the trajectories corresponding to the Nash equilibrium when the exact policies or sufficient demonstrations are given. Our future direction is to identify costs in constrained linear quadratic Gaussian games.

\bibliographystyle{IEEEtran}
\bibliography{main}

\appendices

\section{Proof of Proposition~\ref{prop:characterization}} \label{appendiex:unconstrainedILQG}

The feedback Nash equilibrium of \eqref{problem:unconstrainedLQ} can be found via the coupled backward recursions of the following equations \cite{Basar1998}.
{\small
\begin{align}
&\textstyle K^{i*}_t = \left(R_{t}^{i} + (B_{t}^{i})^\top P^{i*}_{t+1} B_{t}^{i} \right)^{-1} (B_{t}^{i})^\top P^{i*}_{t+1} \left( A_t - \sum_{j \neq i} B_t^j K_t^{j*} \right), \label{recatti:K} \\
&\textstyle \alpha^{i*}_t = \left(R_{t}^{i} + (B_{t}^{i})^\top P^{i*}_{t+1} B_{t}^{i} \right)^{-1} (B_{t}^{i})^\top \left( \zeta^{i*}_{t+1} - P^{i*}_{t+1} \sum_{j \neq i} B_t^j \alpha_t^{j*} \right),  \label{recatti:alpha} \\
&\textstyle P^{i*}_t = F_t^\top P^i_{t+1} F_t + (K^{i*}_t)^\top R_t^i K^{i*}_t  + Q_{t}^{i}, \label{recatti:P} \\
&\textstyle \zeta^{i*}_t = F_t^\top \left( \zeta^{i*}_{t+1} - P^i_{t+1} \sum_{j = 1}^N B_t^j \alpha_t^{j*} \right) + (K^{i*}_t)^\top R_t^i \alpha^{i*}_t + l^i_t, \label{recatti:zeta} \\[-6pt]
& \textstyle P^{i*}_T = Q_{T}^{i}, \quad \zeta^{i*}_T = l_T^i. \notag
\end{align}
}

Given the dynamics, the Nash policy $\{K^{i*}_{t}, \alpha^{i*}_{t}\}_{t \in [T]^-}^{i \in [N]}$, and the cost parameters $\{Q_{t+1}^i, l_{t+1}^i, R_{t}^i\}^{i \in [N]}$ from one step further, we aim to find $\{Q_t^i, l_t^i, R_{t-1}^i\}^{i \in [N]}$ at each $t \in [T]$. We achieve this by rearranging \eqref{recatti:K}-\eqref{recatti:zeta} to construct equations with $Q_t, l_t, R_{t-1}$ being the only unknowns. 

First, we aim to write the value function parameters $P_{t}^{i*}, \zeta_{t}^{i*}$ in terms of $R^i_{t-1}$. Given $F_t = A_t-\sum_{j = 1}^N B_t^j K_t^{j*}$, by rearranging the terms of \eqref{recatti:K}, we get for all $t\in[T]^-$
\begin{align}
&\left(R_{t}^{i} + (B_{t}^{i})^\top P^{i*}_{t+1} B_{t}^{i} \right) \textstyle K^{i*}_t = (B_{t}^{i})^\top P^{i*}_{t+1} \left( A_t - \sum_{j \neq i} B_t^j K_t^{j*} \right), \notag \\
&\Leftrightarrow R_{t}^{i} \textstyle K^{i*}_t = (B_{t}^{i})^\top P^{i*}_{t+1} \left( A_t - \sum_{j = 1}^N B_t^j K_t^{j*} \right), \notag \\
& \Leftrightarrow R_t^i K_t^i = {B_t^i}^\top P_{t+1}^{i*} F_t, \notag\\
&\Leftrightarrow R_t^i {K_t^{i*}} F_t^{-1} = {B_t^i}^\top P_{t+1}^{i*}. \label{eq:BP}
\end{align}

By rearranging the terms of \eqref{recatti:alpha}, we get for all $t\in[T]^-$
\begin{align}
&\textstyle \left(R_{t}^{i} + (B_{t}^{i})^\top P^{i*}_{t+1} B_{t}^{i} \right)\alpha^{i*}_t = (B_{t}^{i})^\top \left( \zeta^{i*}_{t+1} - P^{i*}_{t+1} \sum_{j \neq i} B_t^j \alpha_t^{j*} \right),  \notag \\
&\Leftrightarrow \textstyle \; R_t^i \alpha_t^{i*} = {B_t^i}^\top \left( \zeta_{t+1}^{i*} - P_{t+1}^{i*} \sum_{j=1}^N B_t^j \alpha_t^{j*} \right) \notag\\
& \Leftrightarrow \textstyle \; R_t^i \alpha_t^{i*} + {B_t^i}^\top P_{t+1}^{i*} \sum_{j=1}^N B_t^j \alpha_t^{j*} = {B_t^i}^\top \zeta_{t+1}^{i*}  \notag\\
&\textstyle \stackrel{\eqref{eq:BP}}\Leftrightarrow   \; R_t^i \underbrace{ \textstyle 
 \left[\alpha_t^{i*} + K_t^{i*} F_t^{-1} \sum_{j=1}^N B_t^j \alpha_t^{j*}\right]}_{E_t^i \in \bbR^{n_u}} = {B_t^i}^\top \zeta_{t+1}^{i*}. \label{eq:Bzeta}
\end{align}

For the last step $T$, we have $P^{i*}_T = Q_{T}^{i}$ and $\zeta^{i*}_T = l_T^i$. We aim to replace $P_T^{i*}$ and $\zeta^{i*}_T$ with terms of $R_{T-1}^i$, to construct equations with $Q_T^i, l_T^i$ and $R_{T-1}^i$ being the only unknowns.

Left-multiplying both sides of $P^{i*}_T = Q_{T}^{i}$ by $B_{T-1}^{i^\top}$ gives
\begin{align}
    & B_{t-1}^{i^\top} Q_T^i = B_{T-1}^{i^\top} P_T^{i*} \notag\\
    \stackrel{\eqref{eq:BP}}{\Leftrightarrow} \;& B_{T-1}^{i^\top} Q_T^i = R_{T-1}^i K_{T-1}^{i*} F_{T-1}^{-1}. \notag \\
    \Leftrightarrow \;& B_{T-1}^{i^\top} Q_T^i F_{T-1} - R_{T-1}^i K_{T-1}^{i*} = 0_{n_u \times n_x}. \label{eq:Tqr}
\end{align}

To construct an equation with $l_T^i$ and $R_{T-1}^i$ being the unknowns, we left-multiply $B_{T-1}^{i^\top}$ on both sides of \(\zeta^{i*}_T = l_T^i\),
\begin{align}
&\; B_{T-1}^{i^\top} \zeta_T^{i*} = B_{T-1}^{i^\top} l_T \;\;
\stackrel{\eqref{eq:Bzeta}}\Leftrightarrow \;\; B_{T-1}^{i^\top} l_T - R_{T-1}^i E_{T-1}^i = 0_{n_u \times 1}. \label{eq:Tlr}
\end{align}

Note that \eqref{eq:Tqr} and \eqref{eq:Tlr} are linear equations in $Q_T^i, l_T^i$ and $R_{T-1}^i$. By vectorizing \eqref{eq:Tqr} and \eqref{eq:Tlr}, we get,
\begin{subequations}
\begin{alignat*}{2} 
    & \text{vec}(B_{T-1}^{i \top} Q_T^i F_{T-1}) - \text{vec}\left( R_{T-1}^i K_{T-1}^{i*} \right) = 0_{n_u n_x \times 1}, \\
    & \textstyle \text{vec}(B_{T-1}^{i \top} l_T^i) - \text{vec}\left(R_{T-1}^i E_{T-1}^i\right) = 0_{n_u \times 1}.
\end{alignat*}
\end{subequations}

Using the vectorization identities\footnote{Vectorization properties $\text{vec}(ABC) = (C^\top \otimes A)\text{vec}(B)$ and its extension of $\text{vec}(AB) = (I_n \otimes A)\text{vec}(B) = (B^\top \otimes I_m)\text{vec}(A)$ given $AB = I_m AB = ABI_n \in \bbR^{m \times n}$}, we can write the above equations as $M_T^i \theta_T^i = 0_{n_u(n_x+1) \times 1}$, where \(\theta_T^i = 
[
\textnormal{vec}(Q_{T}^i)^\top, \;
\textnormal{vec}(l_{T}^i)^\top,  \;
\textnormal{vec}(R_{T-1}^{i})^\top,  \;
1]^\top\) and 
{\small
\begin{align*} 
    M_T^i = \begin{bmatrix}
        F_{T-1}^\top \otimes B_{T-1}^{i \top} & 0_{n_un_x \times n_x} & - K_{T-1}^{i* \top} \otimes I_{n_u} & 0_{n_u n_x \times 1} \\
        0_{n_u \times n_x^2} & B_{T-1}^{i \top} &  - E_{T-1}^{i \top} \otimes I_{n_u} &  0_{n_u \times 1}
    \end{bmatrix}.
\end{align*}}

For $t = T-1, T-2, \dots, 1$, we repeat the above approach for $t=T$. Now we know the cost parameters $\{Q_{t+1}^i, l_{t+1}^i, R_{t}^i\}^{i \in [N]}$, from which we have $P_{t+1}^{i*}$ and $\zeta_{t+1}^{i*}$ defined in \eqref{recatti:P} and \eqref{recatti:zeta}, respectively. \vspace{1mm}

To construct an equation with $Q_t^i$ and $R_{t-1}^i$ being the unknowns, we left-multiply $B_{t-1}^{i^\top}$ on both sides of \eqref{recatti:P},
{\small
\begin{align}
    & B_{t-1}^{i^\top} P_t^{i*} = B_{t-1}^{i^\top} [ \underbrace{F_t^\top P_{t+1}^{i*} F_t + {K_t^{i*}}^\top R_t^i K_t^{i*}}_{\Delta^i_t \in \bbR^{n_x \times n_x}} + Q_t^i ] \notag\\
    \stackrel{\eqref{eq:BP}}{\Leftrightarrow} \;& R_{t-1}^i K_{t-1}^{i*} F_{t-1}^{-1} = B_{t-1}^{i^\top} \Delta^i_t + B_{t-1}^{i^\top} Q_t^i \notag\\
    \Leftrightarrow \;& B_{t-1}^{i^\top} Q_t^i F_{t-1} - R_{t-1}^i K_{t-1}^{i*} + B_{t-1}^{i^\top} \Delta^i_t F_{t-1} = 0_{n_u \times n_x}.  \label{eq:BPt}
\end{align}}

To construct an equation with $l_t^i$ and $R_{t-1}^i$ being the unknowns, we left-multiply $B_{t-1}^{i^\top}$ on both sides of \eqref{recatti:zeta},
{\small
\begin{align}
&\; B_{t-1}^{i^\top} \zeta_t^{i*} = B_{t-1}^{i^\top} F_t ( \zeta_{t+1}^{i*} - P_{t+1}^{i*} \sum_{j=1}^N B_t^j \alpha_t^{j*}) \notag \\
& \quad\quad\quad\quad\quad\quad\quad\quad  + B_{t-1}^{i^\top} (K_t^{i*})^\top R_t^i \alpha_t^{i*} + B_{t-1}^{i^\top} l_t \notag\\
\stackrel{\eqref{eq:Bzeta}}\Leftrightarrow &\; B_{t-1}^{i^\top} \underbrace{\left[ F_t ( \zeta_{t+1}^{i*} - P_{t+1}^{i*} \sum_{j=1}^N B_t^j \alpha_t^{j*} ) + (K_t^{i*})^\top R_t^i \alpha_t^{i*} \right]}_{\Omega_t \in \bbR^{n_x}} \notag \\ 
& \quad\quad\quad\quad\quad\quad = R_{t-1}^i E_{t-1}^i - B_{t-1}^{i^\top} l_t \notag\\
\Leftrightarrow &\; B_{t-1}^{i^\top} \Omega_t + B_{t-1}^{i^\top} l_t - R_{t-1}^i E_{t-1}^i = 0_{n_u \times 1}.  \label{eq:Bzetat}
\end{align}}

By vectorizing \eqref{eq:BPt} and \eqref{eq:Bzetat}, which are linear equations in $Q_t^i, l_t^i$ and $R_{t-1}^i$, we get
{\small
\begin{align*}
    & \text{vec}(B_{t-1}^{i \top} \Delta_t^i F_{t-1})  + \text{vec}(B_{t-1}^{i \top} Q_t^i F_{t-1})\\ & \hspace{12.3em}  - \text{vec}\left( R_{t-1}^i K_{t-1}^{i*}\right)  = 0_{n_u n_x \times 1}, \\
    & \text{vec}\left( B_{t-1}^{i \top} \Omega_t^i \right) +  \text{vec}(B_{t-1}^{i \top} l_t^i) - \text{vec}\left(R_{t-1}^i E_{t-1}^i\right) = 0_{n_u \times 1}.
\end{align*}}

Similar to the case of $t=T$, we can write the above equations as \eqref{eq:vec} using the vectorization identity.

\section{Supplementary proof of Theorem~\ref{theorem:costBound}} \label{appendix:exProof}
In this section, we aim to prove \eqref{eq:paramBound} and \eqref{eq:rightImply}. We use $\|\cdot\|$ to denote any vector and matrix norms that are submultiplicative (i.e., $\|AB\| \leq \|A\| \|B\|$), including $\ell_p$ norms ($p \ge 1$) and induced matrix norms. We first present the following result.
\begin{lemma} \label{lemma:blockPresv}
Let $D=[A\;B]\in\mathbb{R}^{p\times (n+m)}$ with $A\in\mathbb{R}^{p\times n}$, $B\in\mathbb{R}^{p\times m}$.
Let us denote
$\widehat A=A+\Delta A$, $\widehat B=B+\Delta B$, $\widehat D=D+\Delta D$ with
$\Delta D=[\Delta A\;\Delta B]$. We have 
\begin{align*}
    \left[\|\Delta A\|= \cO(\epsilon) \wedge \|\Delta B\|= \cO(\epsilon)\right] \Leftrightarrow \|\Delta D\|= \cO(\epsilon).
\end{align*}
\end{lemma} \vspace{0mm}

\noindent
\begin{proof}
    1) The forward direction is straightforward.

2) Reverse direction: let $P_A=\begin{bmatrix}I_n, \; 0\end{bmatrix}^\top\in\mathbb{R}^{(n+m)\times n}$ and
$P_B=\begin{bmatrix} 0, \; I_m\end{bmatrix}^\top\in\mathbb{R}^{(n+m)\times m}$.
Given \(\|\Delta D\|= \cO(\epsilon)\),
\begin{align*}
\|\Delta A\|&=\|\Delta D\,P_A\|\le \|\Delta D\|\,\|P_A\| = \cO(\epsilon), \\
\|\Delta B\|&=\|\Delta D\,P_B\|\le \|\Delta D\|\,\|P_B\| = \cO(\epsilon). \hspace{9mm}\qedhere
\end{align*} 
\end{proof}

From lemma~\ref{lemma:blockPresv}, three conclusions can be drawn as follows.
\begin{enumerate}[label=\alph*)]
    \item Equation \eqref{eq:paramBound} holds with probability $1-\delta$ given \eqref{eq:paramConcentration}; 
    \item If $\|\theta_t^i -\widehat{\theta}_t^i\| = \mathcal{O}(\epsilon)$, then $\|Q_t^i - \widehat{Q}_t^i\| = \mathcal{O}(\epsilon)$, $\|l_t^i - \widehat{l}_t^i\| = \mathcal{O}(\epsilon)$ and $\|R_{t-1}^i - \widehat{R}_{t-1}^i\| = \mathcal{O}(\epsilon)$;
    \item  If all block components of $M_t^i$ admit an $\mathcal{O}(\epsilon)$ error bound, then $\|M_t^i - \widehat{M}_t^i\| = \mathcal{O}(\epsilon)$. 
\end{enumerate}

We will directly use these results in the subsequent proofs.

The matrix $M^i_t$ consists of a left block with components defined in \eqref{eq:defNoBack} that only depends on the Nash policy and system dynamics, and a right block that also depends on $\theta^i_{t+1}$.

{\small \[ M_t^i = \left[
\begin{array}{ccc|ccc} 
        S_{t-1}^i & 0_{n_un_x \times n_x} & - K_{t-1}^{i* \top} \otimes I_{n_u} & S_{t-1}^i\bar{\Delta}_{t}^i \\
        \undermat{\text{Left block}}{0_{n_u \times n_x^2} & B_{t-1}^{i \top} &  - E_{t-1}^{i \top} \otimes I_{n_u} &} \undermat{\text{Right block}}{ B_{t-1}^{i \top}\bar{\Omega}_{t}^i}
    \end{array} \right],
    \]} \vspace{3mm}

To prove \eqref{eq:rightImply}, we show the following two conditions hold for all $t \in [T]^-$:
{\small
\begin{align} \label{eq:left}
   \eqref{eq:paramBound} \Rightarrow \begin{cases}
    \|\widehat{F}_t - F_t\| = \mathcal{O}(\epsilon) \\
    \|\widehat{S}_t^i - S_t^i\| = \mathcal{O}(\epsilon) \\
    \|K_t^{i*\top} \otimes I_{n_u} - \widehat{K}_t^{i*\top} \otimes I_{n_u}\| = \mathcal{O}(\epsilon) \\
    \|\widehat{E}_t^{i \top} \otimes I_{n_u} - E_t^{i \top} \otimes I_{n_u}\| = \mathcal{O}(\epsilon)
    \end{cases}
\end{align} \vspace{-5mm}

\begin{align}\label{eq:right}
 \left[ \eqref{eq:paramBound} \wedge \|\widehat{\theta}_{t+1}^i - \theta_{t+1}^{i}\| = \mathcal{O}(\epsilon) \right] \Rightarrow \begin{cases}
     \|\widehat{S}_{t-1}^i\widehat{\Delta}^{i}_t - S_{t-1}^i \bar\Delta^{i}_t\| = \cO(\epsilon) \\
     \|B_{t-1}^{i \top}\widehat{\Omega}^{i}_t - B_{t-1}^{i \top}\bar{\Omega}_{t}^i\| = \cO(\epsilon)
 \end{cases} 
\end{align}}

If \eqref{eq:paramBound} and $\|\widehat{\theta}_{t+1}^i - \theta_{t+1}^{i}\| = \mathcal{O}(\epsilon)$ hold, \eqref{eq:left} and \eqref{eq:right} together imply $\|M_{t}^i - \widehat M_{t}^i\| = \mathcal{O}(\epsilon)$ by Lemma~\ref{lemma:blockPresv}. Thus, the implication of \eqref{eq:rightImply} is shown.

We present three lemmas that characterize the propagation of small perturbations in matrices through matrix multiplication, Kronecker products, and inverse operations. From these lemmas, \eqref{eq:left} and \eqref{eq:right} readily follow.

\begin{lemma} \label{lemma:product_bound_delta}
Let $\widehat A = A + \Delta A$ and $\widehat B = B + \Delta B$, where $\|\Delta A\| = \mathcal{O}(\epsilon)$ and $\|\Delta B\| = \mathcal{O}(\epsilon)$. For an arbitrarily small $\epsilon$,  \(\|AB - \widehat A \widehat B\| = \mathcal{O}(\epsilon).\)
\end{lemma}

\begin{proof} Substitute $\widehat A = A + \Delta A$ and $\widehat B = B + \Delta B$
\begin{align*}
AB - \widehat A \widehat B
&= AB - (A + \Delta A)(B + \Delta B) \\
&= -A\Delta B - \Delta A B - \Delta A \Delta B.
\end{align*}

Taking norms and by submultiplicativity and the triangle inequality, we have
\[
\|AB - \widehat A \widehat B\|
\le \|A\|\,\|\Delta B\|
+ \|\Delta A\|\,\|B\|
+ \|\Delta A\|\,\|\Delta B\|
= \mathcal{O}(\epsilon). \qedhere
\]
\end{proof}

\begin{lemma} \label{lemma:kronecker_bound_delta}
Let $\widehat A = A + \Delta A$ and $\widehat B = B + \Delta B$, where $\|\Delta A\| = \mathcal{O}(\epsilon)$ and $\|\Delta B\| = \mathcal{O}(\epsilon)$. For an arbitrarily small $\epsilon$,  \(\|A \otimes B - \widehat A \otimes \widehat B\| = \mathcal{O}(\epsilon).\)
\end{lemma}

\begin{proof}
Substitute $\widehat A = A + \Delta A$ and $\widehat B = B + \Delta B$:
\begin{align*}
A \otimes B - \widehat A \otimes \widehat B
=& A \otimes B - (A + \Delta A)\otimes(B + \Delta B) \\
=& -A \otimes \Delta B - \Delta A \otimes B - \Delta A \otimes \Delta B.
\end{align*}

Using the Kronecker product norm property $\|X \otimes Y\| \le \|X\|\|Y\|$ and applying the submultiplicativity and the triangle inequality, we get
\begin{align*}
\|A \otimes B - \widehat A \otimes \widehat B\|
&\le \|A\|\,\|\Delta B\|
+ \|\Delta A\|\,\|B\|
+ \|\Delta A\|\,\|\Delta B\| \\
&= \mathcal{O}(\epsilon).\qedhere
\end{align*}
\end{proof}

\begin{lemma}\label{lemma:inverse_bound}
Let $\widehat F_t = F_t + \Delta F_t$ with $\|\Delta F_t\| = \mathcal{O}(\epsilon)$. For an arbitrarily small $\epsilon$, \(\|\widehat F_t^{-1} - F_t^{-1}\| = \mathcal{O}(\epsilon).\)
\end{lemma}

\begin{proof}
One can verify that, 
\begin{equation}\label{eq:exact-id}
    \widehat F_t^{-1} - F_t^{-1} = -\,F_t^{-1}\,\Delta F_t\,\widehat F_t^{-1}.
\end{equation}

By left-multiplying both sides of \eqref{eq:exact-id} by $F_t$, the equality trivially holds. Given $\|\Delta F_t\|$ is arbitrarily small, we have $\|F_t^{-1}\Delta F_t\|<1$. Thus, the series $(I + F_t^{-1}\Delta F_t)^{-1} = \sum_{k=0}^{\infty}(-F_t^{-1}\Delta F_t)^k$ is convergent and  
\begin{align*}
    \|\widehat F_t^{-1}\| &= \|(F_t +\Delta F_t)^{-1}\| = \|[F_t(I + F_t^{-1}\Delta F_t)]^{-1}\| \\
    & = \| (I + F_t^{-1}\Delta F_t)^{-1} F_t^{-1} \|\\
    & \textstyle\leq  \left( \sum_{k=0}^{\infty}\|-F_t^{-1}\Delta F_t\|^k\right)\|F_t^{-1} \|  \\
    &\le \frac{\|F_t^{-1}\|}{1-\|F_t^{-1}\|\,\|\Delta F_t\|}=\mathcal{O}(1).
\end{align*}

Taking norms of \eqref{eq:exact-id}, we have
\begin{align*}
\|\widehat F_t^{-1} - F_t^{-1}\|
&= \|F_t^{-1}\,\Delta F_t\,\widehat F_t^{-1}\| \\
&\le \|F_t^{-1}\|\, \|\widehat F_t^{-1}\|\, \|\Delta F_t\|  = \cO(\epsilon). \qedhere
\end{align*}
\end{proof}

\section{Sample complexity definition} \label{appendix:nDef}
The sample complexity required for providing the probabilistic error bound \eqref{eq:paramConcentration} for linear regression is presented in \cite[Theorem 3.1]{krikheli2018finitesampleperformancelinear}, we provide it here for completeness.
\begin{align*}
&n(\epsilon, \delta) = \max \left\{ N_1(\epsilon, \delta), \, N_{\text{rand}}(\delta) \right\}, \\
&N_1(\epsilon, \delta) = \frac{8 \alpha^2 \delta^2}{\epsilon^2 \sigma_{\min}^2} \log \left( \frac{2p}{\delta} \right), \\
&N_{\text{rand}}(\delta) = \frac{4}{3} \cdot \frac{(6\sigma_{\max} + \sigma_{\min}) \left( p \alpha^2 + \sigma_{\max} \right)}{\sigma_{\min}^2} \log \left( \frac{2p}{\delta} \right).
\end{align*}
where $\delta$ is the sub-Gaussian parameter of the input observation noise $\nu_t^{(s)}$ and other parameters are defined as follows.
\begin{align*}
&\text{\textPsi} = \frac{1}{n}(\begin{bmatrix}
    X_t \\ \mathbf{1}_{1\times n} 
\end{bmatrix}^\top \begin{bmatrix}
    X_t \\ \mathbf{1}_{1\times n} 
\end{bmatrix}), \\
&\sigma_{\max} := \lambda_{\max}(\text{\textPsi}), \quad \sigma_{\min} := \lambda_{\min}(\text{\textPsi}), \\
&p = \text{rank}(\begin{bmatrix}
    X_t \\ \mathbf{1}_{1\times n} 
\end{bmatrix}^\top \begin{bmatrix}
    X_t \\ \mathbf{1}_{1\times n} 
\end{bmatrix}).
\end{align*}

\section{Sufficient condition for the existence and uniqueness of Nash equilibrium in LQG game} \label{appendix:phi}

A sufficient condition for the existence and uniqueness of the Nash equilibrium of the LQG game \eqref{problem:unconstrainedLQ} is the invertibility of the following matrix for all $t \in [T]^-$ \cite[Remark 6.5]{Basar1998}.

{\small
\begin{align} \label{eq:phi}
\Phi_t =& \begin{bmatrix} 
R_{t}^{1} + (B_{t}^{1})^\top P^{1*}_{t+1} B_{t}^{1}& \dots & (B_{t}^{1})^\top P^{1*}_{t+1} B^N_t\\
(B_{t}^{2})^\top P^{2*}_{t+1} B^1_t & \dots &  (B_{t}^{2})^\top P^{2*}_{t+1} B^N_t\\
\vdots & \ddots & \vdots \\
(B_{t}^{N})^\top P^{N*}_{t+1} B^1_t & \dots &  R_{t}^{N} + (B_{t}^{N})^\top P^{N*}_{t+1} B_{t}^{N}
\end{bmatrix}.
\end{align}}

\end{document}